\documentclass{ifacconf}

\usepackage{graphicx}      
\usepackage{natbib}        
\usepackage{amsmath,amssymb} 

\newcommand \cK{{{\cal{K}}}}
\newcommand \cM{\ensuremath{{\mathcal{M}}}}

\newcommand{\gammainv}{$\gamma-$contracting cone} 
\newcommand\real{\ensuremath{{\mathbb R}}}
\newcommand\realn{\ensuremath{{\mathbb{R}^n}}}

\newcommand\mymatrix[2]{\left[\begin{array}{#1} #2 \end{array}\right]}

\newcommand{\re}{\mathbb{R}}

\providecommand{\rmj[1]}{{\color{red}#1}}
\providecommand{\com[2]}{\begin{tt}[#1: #2]\end{tt}}

\newcommand{\calM}{\mathcal{M}}
\newcommand{\calN}{\mathcal{N}}

\newcommand{\calV}{\mathcal{V}}

\newcommand{\calK}{\mathcal{K}}

\newtheorem{remm}{Remark}

\newenvironment{proof}{\emph{Proof:} }{\hfill \hspace*{1pt} \hfill $\Box$}

\newtheorem{definition}{Definition}
\newtheorem{examp}{Example}
\newenvironment{example}{\begin{examp}\rm }{\hfill \hspace*{1pt} \hfill $\lrcorner$\end{examp}}

\begin{document}
\begin{frontmatter}

\title{Path-complete positivity of switching systems} 

\thanks[footnoteinfo]{
R.J. is currently on sabbatical at UCLA, Department of Electrical Engineering, Los Angeles, USA. His work is supported by the French Community of Belgium and by the IAP network DYSCO.  He is a Fulbright Fellow and a FNRS Fellow.}

\author[First]{F. Forni}, 
\author[Second]{R. M. Jungers},
\author[Third]{R. Sepulchre}

\address[First]{University of Cambridge, United Kingdom \\(f.forni@eng.cam.ac.uk).}
\address[Second]{Universit{\'e} catholique de Louvain, Belgium (raphael.jungers@uclouvain.be)}
\address[Third]{University of Cambridge, United Kingdom (r.sepulchre@eng.cam.ac.uk)}

\begin{abstract}                
The notion of path-complete positivity is introduced as a way to generalize the property of positivity from one LTI system to a family of switched LTI systems whose switching rule is constrained by a finite automaton. The generalization builds upon the analogy between stability and positivity, the former referring to the contraction of a norm, the latter referring to the contraction of a cone (or, equivalently, a projective norm). We motivate and investigate the potential of path-positivity and we propose an algorithm
for the automatic verification of positivity. 
\end{abstract}


\end{frontmatter}

\section{Introduction}

Positivity is a classical concept of linear system theory. It originates in the many examples of system dynamics whose state variables
remain positive along trajectories, and finds its theoretical foundations in Perron-Frobenius theory.  In a nutshell, under mild assumptions, the solutions of a positive system converge
to a dominant eigendirection in the positive orthant \cite{Luenberger1979}. Positivity has known a renewed interest in the recent years
for its advantageous computational scalability over general linear systems \cite{Rantzer2015}. As a geometric concept, positivity is primarily about the contraction
of a cone under the action of a linear map. The positive orthant is a cone of special interest, but Perron-Frobenius theory owes fundamentally
to the geometric contraction of a cone more than to an algebraic property of matrices with positive elements. \\
It is the same contraction property
that makes positivity the infinitesimal (or differential) characterization of monotonicity : the order preserving property of a monotone map
is equivalent to a positivity property  for the linearized map. This geometric viewpoint on positivity is at the root of the differential positivity theory
recently introduced in \cite{Forni2016} to characterize and study the asymptotic properties of nonlinear systems whose trajectories infinitesimally contract a smooth  cone  field. 
It has proven quite insightful to think of differential positivity as an analog of differential stability, or contraction analysis. In one case, one studies
the contraction of a smooth {\it norm} field, e.g. a Riemannian metric, while in the latter case, one studies the infinitesimal contraction of a {\it cone} field. This insight
points to a basic but profound similitude between stability and positivity : two contraction properties, that only differ by the geometric nature of the object that is contracted.

The present paper draws upon this analogy to generalize the concept of positivity from a single matrix (or linear operator) to a family of matrices. Such a generalization
has received considerable attention in the context of stability, but much less in the context of positivity. In particular, we focus in the present paper on the recent 
framework of path-complete Lyapunov analysis, which is a unifying approach to study the stability of a switched system whose switching rule is constrained by a finite automaton. 
Our goal here is to mimick this framework when the \emph{norm contraction} underlying stability is replaced by a \emph{cone contraction} underlying positivity.

The paper is organized as follows. In Section \ref{sec:positivity}, we briefly recall the notion of positivity, and its links with stability.  Then, in Section \ref{sec:contraction}, we naturally draw on this parallel to introduce our main concept: path-complete positivity. In Section \ref{section:behavior} we explain what this concept implies in terms of dynamical systems and control, and finally Section \ref{section:algo} touches upon the algorithmic problem of recognizing this property for a given set of matrices.

\section{Positivity versus stability} \label{sec:positivity}

Both stability and positivity are classical notions
in linear systems analysis. We review basic notations and terminology and 
stress the analogy  between these two properties in the elementary
context of a linear time-invariant (LTI) system $x^+ = Ax$. 

Stability refers to the invariance of a norm, i.e. a {\it ball} in the state-space.
The restriction to {\it quadratic}
norms $|x|_P := \sqrt{x^T P x}$ (where $P$ is a positive definite
matrix) is no loss of generality for LTI systems, in which case the invariance 
condition corresponds to the (Lyapunov)   inequality 
$$
 A^T P A - \gamma P \preceq 0, \; 0 \le \gamma \le 1.
$$
The case $\gamma = 1$ only ensures invariance (i.e. Lyapunov stability) whereas the
case  $\gamma < 1$ ensures contraction (i.e. exponential stability). In Lyapunov
analysis, the norm $V(x) := x^T P x$ is also called a (quadratic) Lyapunov function.

Fundamentally, positivity is the analog property when the ball is replaced by a cone. 
In this paper, a cone $\calK \subseteq \realn$   always means a convex pointed solid cone.
Recall that a set $\calK$ is a convex cone
if $\alpha x + \beta y \in \calK$ for all $x,y \in \calK$ and
all $\alpha,\beta > 0$. Pointed means $-\calK \cap \calK = \{0\}$.
$\calK$ is solid if it  contains $n$ linearly independent vectors.

A linear system is \emph{positive} with respect 
to $\calK$ (in short, $\calK$-positive) if 
\begin{equation*}
 A \calK \subseteq \calK .
\end{equation*}

Positivity only ensures the invariance condition, while
    \emph{strict positivity} also enforces {\it contraction},  by requiring that
the boundary of the cone is mapped
into the interior of the cone
\begin{equation*}
 A \calK \subseteq \mbox{int}\calK .
\end{equation*}

Positivity has a natural metric characterization based on the Hilbert metric
$d_\calK$ associated to the cone $\calK$.
\begin{definition}\cite{Bushell1973}\label{def:hilbert} Given a cone $\cK \in \re^n,$  the corresponding \emph{Hilbert metric} is given by 
\begin{equation*}
d_\calK(x,y) := \log \left(\frac{M_\calK(x|y) }{m_\calK(x|y) }\right)\qquad \forall x,y \in \calK 
\end{equation*}
where
\begin{align*}
M_\calK(x|y) 
& = \inf \{\lambda \,|\, \lambda y - x \in \calK \} 
= \inf \{ \lambda \,|\, \lambda y \in x+ \mathrm{bdr} \calK \}; \\
m_\calK(x|y) 
& = \sup \{ \mu\,|\, x - \mu y \in \calK \} 
=  \sup \{\mu\,|\, \mu y \in  x - \mathrm{bdr} \calK \}.
\end{align*}
We take $M_\calK(x|y) = \infty$ if $\forall \lambda>0,\, \lambda y  \notin x+ \calK$. \end{definition}
The Hilbert metric is in fact a \emph{distance among rays of the cone},
satisfying the property
$d_\calK(\alpha x, \beta y) = d_\calK(x,y)$ for any positive scaling $\alpha$ and $\beta$.  
It is theferore a distance in the projective space.  In short, contraction of a ball is measured by
a norm distance, whereas contraction of a cone is measured by a projective distance.
The Hilbert metric characterizes the contraction of a cone in the same way as a Lyapunov function characterizes
the contraction of a ball, as shown in the following theorem. 
\begin{thm}
\label{bushell}
\cite{Bushell1973} 
Consider a matrix $A\in \re^{n\times n}.$ If $A$ is $\cal K$-positive, then
there exists $\gamma<1$ such that for any $x,y \in \calK$
\begin{equation}
\label{eq:gamma_contraction}
d_\calK(Ax,Ay) \leq \gamma d_\calK(x,y) \ .
\end{equation}
Moreover, the smallest $\gamma$ satisfying the equation above satisfies
$$\gamma =  \tanh \frac{1}{4} D^A_\calK $$ where
$D^A_\calK := \sup_{x,y \in \mathrm{int} \mathcal{K}} d_\calK (Ax,Ay) \ .$
\end{thm}

Clearly $\gamma < 1$ whenever $D^A_\calK < \infty$, that is,
whenever $A  \calK \subseteq \mathrm{int} \calK $.
In what follows we will say that $\calK$ is a
$\gamma$-{\emph{contracting}} cone
for the linear map $A$ whenever \eqref{eq:gamma_contraction} holds.

Proving the contraction of a map is a fundamental way of characterizing the
existence of a fixed point.  Contraction of a ball implies that the iterated map
eventually shrinks to a point. This is the essence of Lyapunov theory. Likewise,
contraction of a cone implies that the iterated map eventually shrinks to a ray (a point
in the projective space). This is the essence of Perron-Frobenius theory. 

For a LTI system, both stability and positivity have a spectral characterization. 
Exponential stability (or contraction) means that all the eigenvalues have a strictly negative
real part, while strict positivity (or projective contraction) means that the matrix $A$ has a
dominant eigenvector in the interior of the cone.

\section{Contraction and path-contraction} \label{sec:contraction}

There exists an extensive literature devoted to generalizing the stability of a single matrix (in the
sense recalled in the previous section) to a finite (or even compact) set of matrices
$A_\sigma \in \real^{n\times n}$,
$\sigma \in \Sigma := \{1,\dots,N\} \subset \mathbb{N}$. See for instance \cite{Liberzon2003,Jungers2009} and references therein.
One obvious application is the stability analysis
of switched systems $x^+ = A_\sigma x$ where the update rule is allowed to switch among the considered set of matrices.
Drawing upon the analogy stressed above, this section generalizes positivity to a set of matrices.

\subsection{Uniform positivity}

A straightforward extension with respect to the previous section is to study uniform positivity
of a family of matrices with respect to a common cone.

Not surprisingly, strict positivity of each matrix (possibly with respect to different cones)
is necessary but not sufficient for uniform strict positivity. And  proving the existence of a 
common invariant cone for a set of linear dynamics is hard. Actually, the existence question is algorithmically
undecidable \cite{Protasov2010,Protasov2012,Rodman2010}, very much for the same reasons as
its companion question of uniform norm contraction (see \cite{BlTsTBOA}, or \cite[Section 2.2.3]{Jungers2009}).

It would certainly be of interest to revisit the large body of literature on uniform stability  in the light of the
 analog question of uniform positivity.
Even the question of defining a joint projective radius for a family of positive systems in analogy to the
`joint spectral radius' defined for a family of stable systems seems valuable and not entirely straightforward.
We do not pursue this question in the present paper and leave it for future research.

\subsection{Constrained switching systems}

Uniform positivity or uniform stability is too conservative of a property
for the many applications where the  switching rule is not arbitrary.
This has long been acknowledged in the literature of switched systems,
see for instance  \cite{EsPhTMAS,BlFeSAOD,LiAnSASO}, where the 
permissible sequence of switches is typically modeled by a finite automaton. 

Consider a class of switching linear systems 
represented by 
\begin{equation} \label{eq-ss}
x(k+1) = A_{\sigma(k)} x(k)  
\end{equation}
where $\sigma \in \Sigma := \{1,\dots,N\} \subset \mathbb{N}$
and each $A_\sigma$ is a $n\times n$ matrix.
For a switching signal $\sigma(\cdot):\mathbb{N} \to \Sigma$ and any initial condition $x_0\in \re^n,$ the unique solution
$x(\cdot): \mathbb{N} \to \realn $ of \eqref{eq-ss}
is called a \emph{trajectory of }the system.
We say that the system is a \emph{constrained switching system}
if the sequences $\sigma(0)\sigma(1)\dots$ generated by
the switching signal $\sigma(\cdot)$ belongs to a regular language $L_r$. 

Thus, $\sigma(\cdot)$ is generated by any finite-state automaton
$(Q,\Sigma,\delta)$ that accepts the same regular language $L_r$, 
where $Q$ is the set of states,
$\Sigma$ is the alphabet and $\delta\subseteq Q\times\Sigma\times Q$ is the 
transition relation.
We say that  such an automaton is \emph{path-complete} 
to emphasize the fact that its paths capture a complete
description of the allowed behaviors of the switching signal.
We will denote any labeled transition 
by the compact notation $i\!\stackrel{\sigma}{\to}\!j\in\delta$.
A finite sequence of transitions from $i$ to $j$ will be represented by
$i\stackrel{\sigma_1\dots\sigma_r}{\longrightarrow}j$.

The complexity of the switching behavior is modulated 
by the automaton. 
An example is in Figure \ref{fig:automata}.
\emph{Arbitrary switches} between two matrices $A_0$ and $A_1$ 
are easily captured by the automaton  on the left.
In contrast, the automaton on the right
enforces a switching behavior with a strict alternation between $0$ and $1$.
A mixed situation is provided by the automaton in the middle,
whose switches sequences allow for any repetition of $1$ separated by 
isolated zeros.

\begin{figure}[htbp]
\begin{center}
\includegraphics[width=0.85\columnwidth]{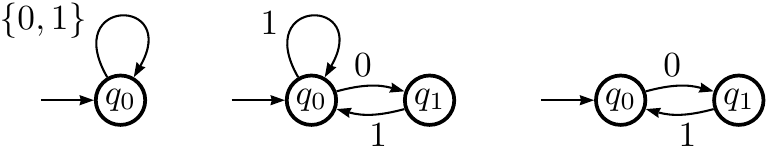}  
\vspace{-2mm}
\caption{Automata with different path restrictions.} 
\label{fig:automata}
\end{center}
\end{figure}

The case of unconstrained switches is typical of robust analysis
where parametric uncertainties are modeled via nondeterministic
switches among a family of  linear systems \cite{Liberzon2003}.
Constrained switches arise from literature on hybrid/cyber-physical 
systems \cite{EsPhTMAS,BlFeSAOD,LiAnSASO}. 
In constrained switching systems, specific sequences of operations 
are captured by suitable branches of the automaton.
Restrictions
on paths could be used to model forms of ergodicity in the 
sequence of matrix operations, or to model the alternation 
between periods of local/isolated operations and periods of 
collective computations.

\subsection{Path-complete Lyapunov functions}
 
Since the nineties, several methods have been proposed for the stability analysis 
of switched systems with or without restrictions on the switching rules 
 \cite{BlFeSAOD,DaRiSAAC,EsLeCOLS,branicky1998multiple}.
We briefly summarize the recently proposed framework  of path-complete Lyapunov functions, \cite{AhJuJSRA}, that provides a unifying 
approach, and generalizes these techniques.

\begin{definition} 
Consider a constrained switching system and let $(Q,\Sigma,\delta)$ be
any path-complete automaton.
A \emph{path-complete Lyapunov function} is a multiple 
Lyapunov function given by a finite set of homogeneous 
positive definite functions $(V_i)_{i \in Q}$, 
$V_i : \re^n \mapsto \re^+$, such that  
$$
V_j(A_\sigma x)\leq \gamma V_i(x). 
$$
for each transition $i\!\stackrel{\sigma}{\to}\!j \in \delta$ 
and each $x\ \in \realn$.
\end{definition}

 The reason of this definition lies in the following theorem.
\begin{thm}[\cite{AhJuJSRA}]\label{thm-pclyap} $ $\\ 
Consider a constrained switching system and let $(Q,\Sigma,\delta)$ be
any path-complete automaton.
The existence of a \emph{path-complete Lyapunov function}
for $\gamma = 1$ is a valid criterion for the stability of the switching system.  
Asymptotic stability requires $0 \leq \gamma < 1$.
\end{thm}

We remark that the regular language that constrains the switches
of a constrained switching system can be generated by infinitely
many automata and each of these provides a different set of path-complete
Lyapunov functions. The selection of a suitable automaton is a degree
of freedom in path-complete analysis. The number of states of the
automaton allows to balance the complexity of the multiple Lyapunov function with the computational efficiency.

\subsection{Path-complete positivity}

We follow the approach of path-complete Lyapunov functions to define the corresponding
notion for positive systems. Once again, the key step is to substitute cones to norms.

\begin{definition}
Consider a constrained switching system and let $(Q,\Sigma,\delta)$ be
any path-complete automaton for this system.
The constrained switching system is \emph{path-complete positive}
with respect to the set of cones  
$$
\overline{\calK} := \{\calK_q\,|\, q \in Q\}
$$
if\footnote{Each $\calK_i$ is a pointed, convex, solid cone.}
\begin{equation*}
A_\sigma \calK_i \subseteq \calK_j
\end{equation*}
for each transition $i \stackrel{\sigma}{\to} j \in \delta$. 
\emph{Strict path-complete positivity} further requires that 
\begin{equation*}
A_\sigma \calK_i \subseteq \mbox{int}\calK_j
\end{equation*} 
for each transition $i \stackrel{\sigma}{\to} j \in \delta$.
\end{definition}

The definition above reduces to positivity when 
each cone in the set $\overline{\calK}$ is identical.
Path-complete positivity is a proper generalization of positivity:
Example \ref{example:pos_vs_pathpos} below discusses the case of a path-complete positive switching system 
that cannot be positive with respect to a common cone.

\begin{example}
\label{example:pos_vs_pathpos}
Consider the constrained switching system $x^+ = A_\sigma x$ with
\begin{equation*}
A_0 = \mymatrix{cc}{5 & 0 \\ 0 & 1} 
\qquad
A_1 = \mymatrix{cc}{1 & 0 \\ 0 & 3} \ ; 
\end{equation*}
and suppose that the automaton in Figure \ref{fig:automata_example_pvpp} is
path complete.
\begin{figure}[htbp]
\begin{center}
\includegraphics[width=0.25\columnwidth]{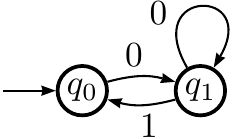}  
\end{center}
\vspace{-2mm}
\caption{One of the automata generating 
$\sigma$ in Example \ref{example:pos_vs_pathpos}.} 
\label{fig:automata_example_pvpp}
\end{figure}

The system cannot be strictly positive with respect to a common cone 
since the dominant eigenvector $e_1$ of the matrix $A_0$ 
is a non-dominant eigenvector of the other matrix $A_1$ and viceversa.
It turns out that the system is strictly path-complete positive
with respect to the family of cones
$\bar{\calK} := \{\calK_0, \calK_1\}$ where
\begin{align*}
\calK_0 & := \{ x_1 \geq 0, |x_2| \leq x_1 \} \\
\calK_1 & := \{ x_1 \geq 0, |x_2| \leq x_1/4\} \ .
\end{align*}
One can check that the path-complete inclusions are satisfied with such values of $\cK_0,\cK_1:$ Indeed, following the automaton paths, 
any $(x_1,x_2) \in \calK_0$ is mapped into
$(x_1^+, x_2^+) \in \mathrm{int}\calK_1$ by $A_0$:
$
A_0 \calK_0 
= \big\{x_1^+ = 5x_1 \geq 0 , |x_2^+| = |x_2| \leq x_1 =x_1^+/5 < x_1^+/4\big\}  
\subseteq \mathrm{int} \calK_1
$.
In a similar way, any $(x_1,x_2) \in \calK_1$ 
is mapped into $(x_1^+, x_2^+) \in \mathrm{int} \calK_1$
by $A_0$:
$
A_0 \calK_1 
 = \{x_1^+ = 5x_1 \geq 0 , 
|x_2^+| = | x_2| \leq x_1/4  = x_1^+/20< x_1^+/4\} 
\subseteq \mathrm{int} \calK_1
$;
and any $(x_1,x_2) \in \calK_1$ 
is mapped into $(x_1^+, x_2^+) \in \mathrm{int} \calK_0$
by $A_1$:
$
A_1 \calK_1 
 = \{x_1^+ = x_1 \geq 0 , 
|x_2^+| = |3 x_2| \leq \frac{3}{4}x_1  = \frac{3}{4}x_1^+< x_1^+\} 
\subseteq \mathrm{int} \calK_0
$.
\end{example}

The definition of path-complete positivity suggests that one of the advantages of 
path-complete positivity is that a temporary ``excess'' of contraction 
can be ``stored'' by narrowing cones. A temporary  ``lack'' of 
contraction can be  ``balanced'' by widening cones.
The example shows that this approach can be effective on 
finite paths: weak contraction at some steps is balanced 
by the excess of contraction at some other steps.

\section{The asymptotic behavior of path-positive systems} 
\label{section:behavior}

The connection between positivity and projective contraction
of the Hilbert metric is now generalized to path-complete positive systems.

\begin{thm}
\label{thm:projective_contraction}
Consider a constrained switching system, let $(Q,\Sigma,\delta)$ be
any path-complete automaton, and suppose that 
the constrained switching system is \emph{path positive}
with respect to the set of cones  
$
\overline{\calK} := \{\calK_q\,|\, q \in Q\}
$.
Then, there exists $0 \leq \gamma \leq 1$ such that, for any 
transition $i \stackrel{\sigma}{\to}j$ of the automaton,
\begin{equation} \label{pcp-gamma}
d_{\calK_j}(A_\sigma x,A_\sigma y) \leq \gamma d_{\calK_i}(x, y)  \qquad \forall x,y \in \calK_i.
\end{equation}
Furthermore, \emph{strict path positivity} guarantees $0 \leq \gamma < 1$.
\end{thm}
\begin{proof}
Following the proof argument for Theorem 3.1 in \cite{Bushell1973}, one shows that
path positivity guarantees
$
m_{\calK_i}(x|y) \leq 
m_{\calK_j}(A_\sigma x|A_\sigma y) \leq
M_{\calK_j}(A_\sigma x|A_\sigma y) \leq 
M_{\calK_i}(x|y) 
$
for each transition $i \stackrel{\sigma}{\to} j \in \delta$, 
which directly implies \eqref{pcp-gamma} for $\gamma = 1$. 

 For strict path positivity
\eqref{pcp-gamma} with $0 \leq \gamma < 1$ follows
by the proof argument of Theorem 3.2 in \cite{Bushell1973}.
For instance, for each $q\in Q$, define the oscillation
$\mathrm{osc}_{\calK_q}(x|y) := M_{\calK_q}(x|y) - m_{\calK_q}(x|y)$.
Theorems 4 and 5 in \cite{Bauer1965} show that
$\mathrm{osc}_{\calK_j}(A_\sigma x|A_\sigma y) \leq N_{ij}(A_\sigma) \mathrm{osc}_{\calK_i}(x|y)$
for each $i \stackrel{\sigma}{\to} j \in \delta$, where the oscillation ratio $0 \leq N_{ij}(A_\sigma) < 1$
if $A_\sigma \calK_i \subseteq \mathrm{int} \calK_j$.
This result is well known for positive operators from a cone into itself. 
The proof argument in \cite{Bauer1965} 
extends to the case of positive operators between two different cones. 
Finally, using the proof argument of Lemma 3 in \cite{Bushell1973a} one shows that
$d_{\calK_j}(A_\sigma x|A_\sigma y) \leq N_{ij}(A_\sigma) d_{\calK_i}(x|y) $
for each $i \stackrel{\sigma}{\to} j \in \delta$.
Thus, $\gamma := \max\nolimits\limits_{ i \stackrel{\sigma}{\to} j \in \delta} N_{ij}(A_\sigma) < 1$.
\end{proof}

At each transition $i\stackrel{\sigma}{\to} j$ strict positivity guarantees that
the linear map $A_\sigma$ is a contraction on the rays of the cones,
in the sense of the adapted Hilbert metrics $d_{\calK_i}$. It is easy to prove, by induction, that
any pair $(x(\cdot)$,$y(\cdot))$ of trajectories of the system associated
to the same switching signal $\sigma(\cdot)$ and such that
$x(0),y(0) \in \calK_{q(0)}\setminus\{0\}$ satisfy
\begin{equation}
\label{eq:contraction}
\lim_{k \to \infty} \left| \frac{x(k)}{|x(k)|} - \frac{y(k)}{|y(k)|} \right| = 0 \ .
\end{equation}

Equation \eqref{eq:contraction} makes clear that a strictly positive system 
asymptotically `forgets' its initial condition,
as it converges to a unique steady 
state solution in the projective space, for every switching signal. 

Note that the projective contraction property does not enforce convergence to a
fixed point. For example, a straightforward consequence of the theorem is 
that each cyclic path $q\stackrel{\sigma_1\dots\sigma_r}{\longrightarrow}q $
defines a corresponding  path-dependent Perron-Frobenius eigenvalue and eigenvector,
$\lambda_{{\sigma_1  \dots  \sigma_k}}$ and $v_{{\sigma_1  \dots  \sigma_r}}$,
such that  
$$
A_{\sigma_r} \dots A_{\sigma_1} v_{{\sigma_1  \dots  \sigma_r}}
=
\lambda_{{\sigma_1  \dots  \sigma_k}} v_{{\sigma_1  \dots  \sigma_r}}
$$
(since $\bar{A} := A_{\sigma_1},\dots,A_{\sigma_r}$
is necessarily a strictly positive matrix).
Denoting rays by
$[x] := \{\lambda x \,|\, \lambda > 0\}$, 
a simple permutation of indices shows that 
$
[v_{\sigma_2  \dots  \sigma_r\sigma_1}] = [A_{\sigma_1} v_{\sigma_1  \dots \sigma_r}]
$,
$
[v_{\sigma_3  \dots  \sigma_r\sigma_1\sigma_2}] = [A_{\sigma_2} v_{\sigma_2  \dots  \sigma_r\sigma_1}]
$, and so on. Indeed, all the path-dependent Perron-Frobenius eigenvectors on a cyclic path define
an invariant sequence of rays. Such sequence is also an attractor of the system.
Thus, trajectories along these cycles
either converge to zero or to a limit cycle of $r$ rays. In that sense, path-positivity retains the fundamental
contraction property of a positive system.

\section{Algorithms for deciding positivity}\label{section:algo}

Testing the existence of a common invariant or contractive cone is hard \cite{Protasov2010}. In fact Protasov proved that the question of whether a set of matrices has an invariant cone is Turing-undecidable. His construction suggests that the question is hard when the matrices share a common invariant linear subspace. For matrices that do not share a common invariant subspace, we algorithmically test whether a given set of matrices has a common $\gamma$-contracting cone, for a given contraction ratio $0 < \gamma < 1$. 
We only discuss the algorithm in the case of uniform
positivity and leave for future work  a generalization to path-complete positivity.

\subsection{Basic test}

A single matrix admits a contracting cone if and only if it has a leading eigenvector.
An obvious necessary condition for uniform strict positivity w.r.t. a common cone $\calK$ is therefore
that each system $A_\sigma$  has a leading eigenvector.
We introduce a  corresponding splitting of the state-space, which relies on the eigenstructure of $A_\sigma.$
\begin{definition}\label{def-splitting} 
For any positive matrix $A_\sigma,$  we define the \emph{invariant splitting} of $\re^n$ $(\calV_\sigma,\calN_\sigma)$ as the pair of two \emph{$A_\sigma-$invariant subspaces} of dimension $1$ and $n-1$ respectively.
$\calV_\sigma$ is defined as the span of the Perron Frobenius eigenvector of $A_\sigma$. 
$\calN_\sigma$ is the unique $n-1$ invariant subspace for $A_\sigma$ such that $\calV_\sigma \cap \calN_\sigma = \{0\}$
(for example $\calN_\sigma$ could be defined by the columns of the coordinate
transformation that brings $A_\sigma$ into its real Jordan form). 
\end{definition}An elementary necessary condition is as follows.
\begin{prop} \label{prop:plane}
If a cone $\cK$ is invariant for the matrix $A_\sigma$, then necessarily  $$ \cK \bigcap \calN_\sigma = \{0\}. $$
\end{prop}

\begin{cor}[Basic test]
If a set of matrices $\cM=\{A_\sigma\}$ share a common contracting cone, then they all have a strictly dominant eigenvalue, and the corresponding eigenvector $v_\sigma$ does not belong to any $\calN_{\sigma'}$ for any $\sigma'\neq \sigma.$
\end{cor}

 \subsection{Inner bound}
 
The basic idea of our algorithm below is to start from an inner bound, and proceed by forward iteration (i.e. apply our matrices to this inner bound) in order to enlarge it. 
For the initial inner bound, one can start with the convex hull of the leading eigenvectors of the matrices, which must be in any invariant cone. Given the set of leading eigenvectors $\{v_i\},$ it is not clear however whether to use $v_i$ or $-v_i$ in the initial inner bound.  We resolve this choice as follows: pick any matrix $A_\sigma$ and define $w$ as the normal vector to the invariant subspace   $\calN_\sigma$. Then for each leading eigenvector of the matrices $A_i$, pick the orientation $v_i$ such that $v_i^Tw>0.$
We formalize the argument in the following proposition.
\begin{prop}[The orientation trick] \label{prop-orientation} 
Suppose that\\ $A_1,A_2\in \re^{n\times n}$ have a common contracting cone $\cK,$ and note $v_1,v_2$ the leading eigenvectors of $A_1,A_2.$  Suppose  without loss of generality that $v_1 \in \cK.$  Then, with the notations of Definition \ref{def-splitting},  $v_2$ is also in $\cK$ if and only if $(w_1^Tv_1)\cdot (w_1^Tv_2)>0,$ where $w_1$ is the normal vector to $\calN_1.$   
\end{prop}
\begin{proof} If: one has either $v_2\in \cK,$ or $-v_2 \in \cK.$  Now, if $(w_1^Tv_1)\cdot (w_1^Tv_2)>0,$ it means that $(w_1^Tv_1)\cdot (w_1^T(-v_2))<0,$ and then there exist $\alpha,\beta >0$ such that  $\alpha v_1+\beta (-v_2)\in \calN_1,$ and thus $(-v_2)$ cannot belong to $\cK.$\\ Only if: Suppose by contradiction that 
$(w_1^Tv_1)\cdot (w_1^Tv_2)<0.$  Then there exist $\alpha,\beta >0$ such that  $\alpha v_1+\beta v_2 \in \calN_1,$ and this contradicts $v_1,v_2$ being in $\cK,$ because $\cK\bigcap \calN_1 = \{0\}.$\end{proof}

By construction, the convex hull of leading eigenvectors (selected with the proper orientation) is an invariant cone. It thus provides an inner bound for the contracting cone $\cK$. However, the following proposition shows that this cone cannot be contracting, even if there exists a contracting cone. 
\begin{prop} \label{prop-nonstrict}Consider a set of matrices $\cM \in \re^{n \times n},$ and the set of leading eigenvectors $v_i \in \re^n$ of the matrices in $\cM.$ Suppose that
$$\cK_0=\mbox{conic-hull }{\bigcup_{A\in \cM^*,v_i}{A v_i}} $$ is a closed convex pointed cone\footnote{
$\calM^*$ is the set of all the products of matrices of $\calM$.}.  Then,  $\cK_0$ is an invariant cone, but not a \gammainv{} for any $\gamma <1.$ 
\end{prop}
\begin{proof}
It is obvious that $\cK_0$ is invariant by definition of $\cK_0.$  Now, let us suppose by contradiction that $\cM \cK_0 \subset \mbox{int } \cK_0.$ Since every $v_i$ is an eigenvector of some matrix $A\in \cM,$  $\mbox{conic-hull }\{v_i\}$ is not contracting. Thus, there exists some $x^* \in \cK_0\setminus \mbox{int }  \cK_0,$ $x^*\notin \mbox{conic-hull }\{v_i\}.$  By definition of $\cK_0,$ for any $\epsilon>0,$ there is a $x'\in \cK_0,\ A\in \cM$ such that $|Ax'-x^*|<\epsilon.$  This is in contradiction with the fact that $\cM \cK_0 \subset \mbox{int } \cK_0,$   $\cM \cK_0 $ being a finite union of closed sets strictly contained in $\cK_0.$
\end{proof}

\begin{example}Consider the set of matrices $\cM=\{A_1,A_2\}.$

$$
A_1
=
\mymatrix{cc}{
    2  &       0 \\
    1.65   & 0.5 \\
}
\qquad
A_2
=
\mymatrix{cc}{
    2    &    0 \\
    1.3636  &  0.5\\
}.
$$
 The leading eigenvectors are $v_1=[1.1, 1]^T$ and $v_2=[1 , 1.1]^T.$ 
$\cM$ possesses an invariant cone, which is $$\mbox{conic-hull }\{v_i\}.$$
However, the cone $\cK_\epsilon =\mbox{conic-hull }\{[\epsilon ,1]^T,[1,0]^T\},$ for $\epsilon >0$ small enough, is contracting (take for instance $\epsilon =0.1$). If an algorithm proceeds by forward iteration starting from $\mbox{conic-hull }\{v_i\},$ it remains stuck in the cone delimited by these two vectors: $$\cK_0=\mbox{conic-hull }{\bigcup_{A\in \cM^*,v_i}{A v_i}}= \mbox{conic-hull }{ \{v_i\}}. $$ That is, $\cK_0$ is an invariant cone, but not contracting. It is however included in the contracting cone $\cK_\epsilon.$
\end{example}

\subsection{Algorithm with guaranteed termination and accuracy} \label{sec:RJ_alg}

In any contracting cone, Theorem \ref{bushell}  implies a uniform upper bound
$D^A_\calK$ on the distance between two points in $A\cK.$ This bound on the distance  is useful to build a \emph{larger inner bound on the contracting cone $\cK$}: indeed, an upper bound on the distance between two points translates geometrically into a lower bound on the distance between any of these points and the boundary of the considered cone (see Definition \ref{def:hilbert} of the Hilbert metric, and the proof of Lemma \ref{lem:inflate} below).  Thus, we can leverage this information in order to inflate the cone, by `pushing the boundaries' of our inner bound.  We formalize this in the next lemma:

\begin{lem}
\label{lem:inflate}
Let $\calK$ be a $\gamma$-contracting cone for a set of linear maps $\calM,$ and take a matrix $A\in \cM.$ 
For any  $x,y \in \calK$, consider $x'=Ax,\ y'=Ay.$  Suppose that $x'-y' \notin \calK;$ then, for any 
$\rho \geq \exp(D^A_\calK)$ we have that
$$y' + \frac{1}{\rho - 1}(y'-x') \in \calK.$$
\end{lem}
\begin{proof}  From Definition \ref{def:hilbert} we have $0 \leq m(Ax|Ay) < 1$ since $Ax-Ay \notin \calK$.
Thus,
\begin{eqnarray*} 
M(Ax|Ay) &=& \frac{M(Ax|Ay)}{m(Ax|Ay)} m(Ax|Ay)\\ &\leq& \exp(D^A_\calK) m(Ax|Ay) \\&\leq &\exp(D^A_\calK).
\end{eqnarray*}
Furthermore,
$M(Ax|Ay)Ay - Ax \in \calK$, thus
$\rho Ay - Ax \in  \calK$
since $\rho \geq M(Ax|Ay).$
Finally, writing
$\rho A y - Ax + A y - Ay \in \calK,$ we obtain
$Ay + \frac{1}{\rho - 1}(Ay-Ax) \in \calK.$
\end{proof}

Lemma \ref{lem:inflate} provides a way to widen any inner bound of $\calK$ in such
a way that the widened cone is still a subset of $\calK$. Indeed, if an inner bound is not contractive, we can use the lemma to widen its boundary slightly outwards before pursuing the forward iteration algorithm.  
Lemma \ref{lem:inflate} and Lemma \ref{lem:lambda} below are at the core of 
Algorithm \ref{algo-cone}, which decides in finite time whether a given set of matrices 
has a common $\gamma$-contracting cone, as clarified in Theorem \ref{thm-algo}.

\begin{lem}
\label{lem:lambda}
Let $\calK$ be a cone in $\re^n,$ $x,y \in \mbox{int } \calK$, and suppose that $d_\calK(x,y)>0.$ Then, for any $(n-1)-$dimensional hyperplane $H$ such that $H\bigcap \cK=\{0\},$ there exists a $\lambda>0$ such that $y-\lambda x \in H.$
\end{lem}
\begin{proof}
For $\lambda$ very small, we have $y-\lambda x \in \cK;$ for  $\lambda$ very large, we have $y-\lambda x \in -\cK.$  Thus, by continuity, there must be a $\lambda$ such that  $y-\lambda x \in H.$
\end{proof}

In the next theorem, we suppose that the matrices do not have zero eigenvalues, nor common invariant subspace.  
These are technical assumptions that hold for generic matrices. 

\begin{thm}\label{thm-algo} Consider a set of positive matrices $\cM$ with nonzero eigenvalues and no common invariant subspace. Given a contraction ratio $\gamma$, 
Algorithm \ref{algo-cone} decides in finite time whether the set of matrices has a common $\gamma$-contracting cone.
\begin{itemize}
\item
It returns a \gammainv{} provided that such a cone exists.
\item 
If there is no \gammainv{}, it returns `NO', or a $\delta-$invariant cone, for $\gamma<\delta <1,$ if it has found one.
\end{itemize}
\end{thm}

\begin{algorithm}[htbp]
\begin{center}
\includegraphics[width=0.9\columnwidth]{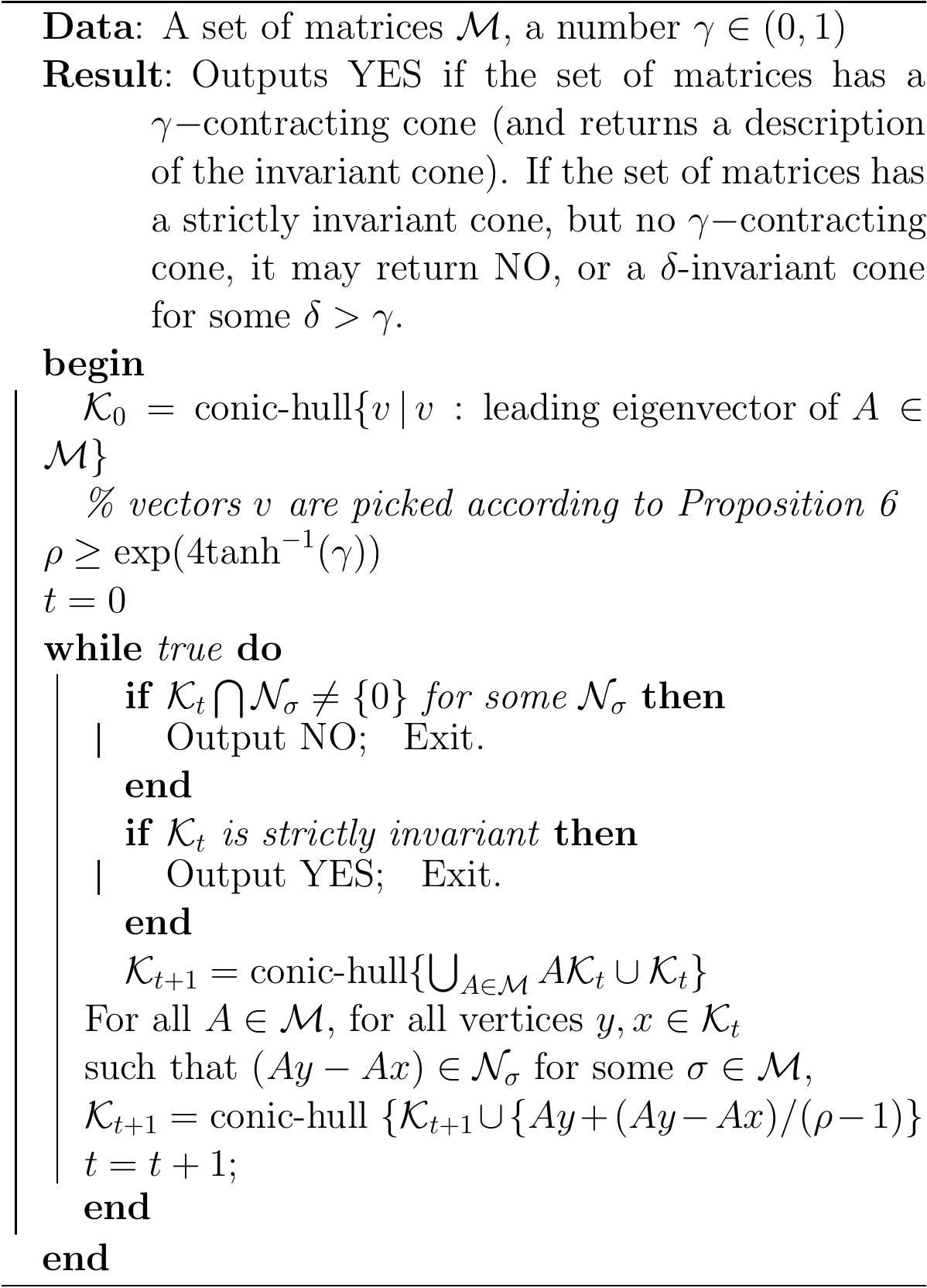}  
\end{center}
\caption{An algorithm for deciding joint positivity}
\label{algo-cone}
\end{algorithm}

\begin{proof}
The algorithm iteratively computes inner bounds $\calK_t$ for $\cK.$ We first prove that indeed $\calK_t$ are valid inner bounds (provided that there indeed exists a contracting cone $\calK$).  We then prove that one of these inner bounds $\calK_t$ must be contracting for some $t$ (and not only invariant).  Thus, the algorithm will terminate with an effective contracting cone.

The algorithm starts with $\cK_0$ as a first inner bound (if $\cK_0$ is not a convex pointed cone, then one directly concludes that the set of matrices does not have a common invariant cone).  Thus, suppose that $\cK_0$ is a valid inner bound by Proposition \ref{prop-nonstrict}.  Now, at every step, with an initial inner bound $\cK_t,$ the algorithm performs two operations. First, it takes the union of $\cK_t$ with all the images of this set $$\calK_{t+1} = \mbox{conic-hull}\{ \bigcup_{A\in\calM} A \calK_t \cup \calK_t \},$$ which is clearly still an inner bound, by definition of a contracting cone.  Then, for any two points $x,y\in \cK_{t},$ it adds the point $\{Ay+(Ay-Ax)/(\rho-1)\} $ to $\cK_{t+1}.$ \\ For practical efficiency, the true algorithm can only do it for points $Ay,Ax$ that are vertices of the new inner bound.  Also, one has to scale $x$ in order to ensure the condition $(Ay - Ax) \in \calN_\sigma,$ but, provided that $Ax$ and $Ay$ are not parallel, this is always possible by Lemma \ref{lem:lambda} above.  Finally note that there must always be an $Ay\in \cK_{t+1} \setminus \mbox{int} \cK_{t+1}$ (if not, $\cK_{t+1}$ would be contracting) and $Ax$ non-aligned with $Ay$ (because the matrix $A$ has nonzero eigenvalues).   In turn, the condition $(Ay - Ax) \in \calN_\sigma$ implies that $Ay-Ax \notin \cK$ (Proposition \ref{prop:plane}), and we can apply Lemma \ref{lem:inflate}.  These new added points are guaranteed to be in the invariant cone $\cK,$ by Lemma \ref{lem:inflate}, and this proves that $\cK_{t+1}$ is still an inner bound.\\
We now prove that this procedure generates a contracting cone after a finite number of steps (if there exists one).  Suppose the contrary.  Then, the inner bounds $\cK_t$ converge towards a cone $\cK_{\infty}$ which is invariant, but not contracting. 
Consider a vertex $z$ of $\cK_\infty,$ which is such that $Az \in \cK_\infty \setminus \mbox{int }\cK_\infty.$  That is, $Az$ is in the boundary of $\cK_{\infty}.$  This implies that the inflating step in the algorithm is such that for all $ x\in \cK_\infty,$ $Ax$ is in the same face of $\cK_\infty$ as $Az$ (because in the opposite case, the inflation step would `push' $Az$ out of $\cK_\infty,$ and $Az$ would not be a vertex anymore).  $\cK_\infty$ being of nonempty interior (because the matrices have no nontrivial invariant subspace), this implies that $A$ has zero eigenvalues, a contradiction.  

In conclusion, the algorithm cannot converge to a non contracting invariant cone.  Thus, if there exists a contracting  $\cK,$ since $\cK_t{}$ are valid inner bounds (i.e. contained in $\cK$), they will either converge to $\cK,$ or the algorithm will stop before (having found another contracting cone).  If, on the other hand, there is no invariant cone, the `inner bounds' will keep growing until they intersect some $\calN_\sigma,$ and the algorithm will stop, concluding that there is no \gammainv{}.
\end{proof}



\section{Conclusions and further directions.}

In this work, we have introduced the concept of path-complete positivity, which generalizes the notion of positivity.  We showed that this notion can be useful, for instance for (constrained) switching systems, for which we provide an example of system which is not positive, but yet, is path-complete positive. We showed that path-complete positive systems inherit much of the nice properties of positive systems, and we sketched an algorithm to decide whether a switching system has an invariant cone.\\
Our algorithm is inspired from the similar, and much more studied, problem of proving stability for switching systems.  It proceeds by \emph{forward propagation,} which is a well-known technique for proving stability of a switching system.  However, the positivity problem is more tricky, for several reasons: first, contrary to the stability problem, one cannot take an arbitrary norm for initializing a forward propagation procedure.  Second, the forward iteration converges by essence to a \emph{non contracting} invariant cone, forcing us to introduce an `inflation procedure' in order to generate a contracting cone.

We believe that path-complete positivity opens a number of directions
worth exploring. The first step is to decide path-complete positivity for a switching system, as mentioned above. In a second time, we plan to generalize this notion to more general dynamical systems, and link it to the notion of differential positivity.

Path-complete positivity may also prove useful for computational goals, even for systems that \emph{do have a common contracting cone,} that is, that are positive.  Just like path-complete stability has been used as a proxy to prove stability in the control literature, it might be easier to prove path-complete positivity than to prove positivity.

\end{document}